\documentclass[11pt,letterpaper,oneside,english]{article}
\usepackage[margin=1.4in]{geometry}  

\usepackage{mdwlist}

\usepackage{amssymb,amsbsy,latexsym,amsthm}
\usepackage{amsmath}
\usepackage{graphics, subfigure, float}
\usepackage{fp, calc}
\usepackage{bm}

\usepackage[latin1]{inputenc}
\usepackage[american]{babel}
\usepackage[T1]{fontenc} 

\def\url#1{{\texttt #1}}

\usepackage{bm}

\usepackage{amscd} 

\usepackage{pst-all}
\usepackage{verbatim, comment}

\newtheoremstyle{theorem}{1em}{1em}{\slshape}{0pt}{\bfseries}{.}{ }{}
\theoremstyle{theorem}
\newtheorem{theorem}{Theorem}
\newtheorem*{theorem*}{Theorem}

\newtheorem{lemma}[theorem]{Lemma}

\providecommand{\setN}{\mathbb{N}}

\providecommand{\setR}{\mathbb{R}}

        \def\drawRect#1#2#3#4#5{
           \FPeval{\x2}{(#2) + #4} 
           \FPeval{\y2}{(#3) + #5} 
           \pspolygon[#1](#2,#3)(\x2,#3)(\x2,\y2)(#2,\y2)
        }

\usepackage[displaymath,textmath,graphics, subfigure, floats]{preview} 
\PreviewEnvironment{center} 
\PreviewEnvironment{pspicture} 
 \usepackage{datetime}

\makeatother

\usepackage{calrsfs} 
\DeclareMathAlphabet{\pazocal}{OMS}{zplm}{m}{n}

\title{A $(1+\varepsilon)$-Approximation for Makespan Scheduling with Precedence Constraints using LP Hierarchies\footnote{The conference version of this work appeared in the 48th ACM Symposium on Theory of Computing (STOC 2016).}}
\date{} 

\author{Elaine Levey\thanks{University of Washington, Seattle, USA. Email: elevey@cs.washington.edu} \and Thomas Rothvo{ss}\thanks{University of Washington, Seattle, USA. Email: rothvoss@uw.edu. Supported by NSF grant 1420180 with title ``\emph{Limitations of convex relaxations in combinatorial optimization}'' and an Alfred P. Sloan Research Fellowship.}}

\begin{document}

\maketitle

\begin{abstract}
\noindent 
In a classical problem in scheduling, 
one has $n$ unit size jobs with a \emph{precedence order} and the goal is to find 
a schedule of those jobs on  $m$ identical machines as to minimize the makespan. It is one of the remaining four open problems from the book of Garey \& Johnson whether or not this problem is $\mathbf{NP}$-hard for $m=3$. 
 
We prove that for any fixed $\varepsilon$ and $m$, an LP-hierarchy lift of the time-indexed LP with a slightly super poly-logarithmic number of $r = (\log(n))^{\Theta(\log \log n)}$ rounds provides a $(1 + \varepsilon)$-approximation. 
For example Sherali-Adams suffices as hierarchy. This implies an algorithm that yields a $(1+\varepsilon)$-approximation in time $n^{O(r)}$. 
The previously best approximation algorithms guarantee
a $2 - \frac{7}{3m+1}$-approximation in polynomial time for $m \geq 4$ and $\frac{4}{3}$ for $m=3$. 
Our algorithm is based on a recursive scheduling approach where in each step we reduce the correlation in form of long chains. Our method adds to the rather short list of examples where hierarchies are actually useful to obtain better approximation algorithms. 
\end{abstract}


\section{Introduction}

One of the landmarks in the theory of scheduling is the paper of Graham~\cite{Graham66} from 1966,
dealing with the following problem: suppose we have a set $J$ of $n$ jobs, each one
with a running time $p_j$ along with $m$ identical parallel machines that we can use to
process the jobs. Moreover, the input contains a \emph{precedence order} 
on the jobs; we write $j \prec j'$ if job $j$ has to be completed before job $j'$
can be started. The goal is to schedule the jobs in a non-preemptive fashion so that the
\emph{makespan} is minimized. Here, the makespan gives the time that the last job is finished.
In the \emph{3-field notation}\footnote{In the 3-field notation, the first field specifies the available processors, the 2nd field the jobs and the last field the objective function. In our case, $Pm$ means that we 
have $m$ identical machines; $p_j=1, \textrm{prec}$ indicates that the jobs have unit length and precedence constraints and the last field $C_{\max}$ specifies that the objective function is to minimize the maximum completion time.}, this problem is abbreviated as $P \mid \textrm{prec} \mid C_{\max}$. Graham showed that the following \emph{list schedule} gives a $(2-\frac{1}{m})$-approximation on the makespan: compute an arbitrary topological ordering of the jobs and whenever a machine becomes idle, select the first available job from the list. It had been known since the late 70's that it is $\mathbf{NP}$-hard to approximate the problem better than within a factor of $4/3$ due to Lenstra and Rinnooy Kan~\cite{ComplexityOfScheduling-Lenstra-RinnoyKan1978} and Schuurman and Woeginger~\cite{TenOpenProblems-SchuurmanWoeginger1991} prominently placed the quest for any improvement on their well known list of 10 open problems in scheduling. Finally in 2010, Svensson~\cite{HardnessPrecedenceScheduling-Svensson-STOC2010} showed that assuming a variant of the unique games conjecture~\cite{OptLongCodeTest-BansalKhot-FOCS09}, there is no $(2-\varepsilon)$-approximation algorithm for $P \mid \textrm{prec}, p_j=1 \mid C_{\max}$. 
However, for unit size jobs, Lam and Sethi~\cite{TwoSchedulingAlgorithms-Lam-Sethi-SICOMP77} analyzed an
algorithm of Coffman and Graham and showed that it provides a slighly better guarantee of $2 - \frac{2}{m}$  for $P \mid \textrm{prec},p_j=1 \mid C_{\max}$. Later, Gangal and Ranade~\cite{PrecedenceConstrainedScheduling-GangalRanade2008} gave an algorithm with a $2 - \frac{7}{3m+1}$ guarantee for $m \geq 4$.

In a typical scheduling application, the number of jobs might be huge compared to the number of machines, 
which does justify to ask for the complexity status of such problems if the number $m$ of machines
is a constant. Even under the additional restriction of unit size jobs, no better approximation result
is known. In fact, it is one of the remaining four open problems from the book of 
Garey and Johnson~\cite{GareyJohnson79}
whether $P3 \mid \textrm{prec},p_j = 1 \mid C_{\max}$ is even $\mathbf{NP}$-hard. Also
Schuurman and Woeginger~\cite{TenOpenProblems-SchuurmanWoeginger1991} list under ``Open Problem 1'' the question whether there is a PTAS for this
problem (recall that for $m=2$, the result of \cite{TwoSchedulingAlgorithms-Lam-Sethi-SICOMP77} gives an optimum schedule).

To understand where the lack of progress is coming from, one has to go back to the 
list scheduling algorithm of Graham. If we schedule the jobs in a greedy manner, then one can argue 
that there is always a \emph{chain} of jobs $j_1 \prec j_2 \prec \ldots \prec j_k$ so that at any point in time
either all $m$ machines are fully busy or a job from that chain was processed. Since both quantities, 
the \emph{load} $\frac{1}{m} \sum_{j \in J} p_j$ and the length of any chain are lower bounds on any
schedule, we can conclude that the schedule has length at most $2 \cdot OPT$. 
One can shave off a factor of $\frac{1}{m}$ even for general running times, by observing that the 
processing times of the 
jobs in the longest chain do not need to be again counted in the load bound. 
Also the papers \cite{TwoSchedulingAlgorithms-Lam-Sethi-SICOMP77} and \cite{PrecedenceConstrainedScheduling-GangalRanade2008} effectively rely on those two lower bounds. 
In fact, \cite{TechReportCharikar1995} showed that a large class of algorithms including the
ones of \cite{Graham66,PrecedenceConstrainedScheduling-GangalRanade2008} cannot beat 
a bound of $2 - \frac{2}{\sqrt{m}}$; moreover Graham's algorithm is indeed not better than
a $(2-\frac{2}{m})$-approximation for unit size jobs, see \cite{PrecedenceConstrainedScheduling-GangalRanade2008}.

Of course, one always has the option to study the strength of linear programs for an optimization problem.
The most natural one for $Pm \mid \textrm{prec},p_j=1 \mid C_{\max}$ is certainly the following
\emph{time-indexed LP}: 
For a parameter $T$ that denotes the length of the time horizon, 
we define a set $K(T)$ as the set of fractional solutions to:
\begin{eqnarray}
 \sum_{t=1}^T x_{j,t} &=& 1 \hspace{1.3cm} \forall j \in J    \label{eq:LP} \\
 \sum_{j \in J} x_{j,t} &\leq& m \hspace{1.2cm} \forall t \in [T] \nonumber \\ 
 \sum_{t' < t} x_{i,t'} &\geq& \sum_{t' \leq t} x_{j,t'} \quad \forall i \prec j \; \forall t \in [T] \nonumber \\
  0 \leq x_{j,t} &\leq& 1 \hspace{1.3cm} \forall j \in J \; \forall t \in [T] \nonumber
\end{eqnarray}
Here $x_{j,t}$ is a decision variable that is supposed to tell whether job $j \in J$ is scheduled in 
time slot $t \in [T]$, where $[T] := \{ 1,\ldots,T\}$. The constraints guarantee that in an integral solution each job is assigned
to one time slot; no time slot receives more than $m$ jobs and for a pair of jobs $i \prec j$, job $i$
has to be scheduled before $j$. 

Unsurprisingly, this LP has a constant integrality gap as one can see from the following construction: 
take $k$ blocks $J_1,\ldots,J_k$ of $|J_i| = m+1$ jobs each and define the precedence order so
that all the jobs in $J_i$ have to be finished before any job in $J_{i+1}$ can be started. 
Any integral schedule needs two time units per block, hence $OPT = 2k$. On the other hand, the LP
solution can schedule the $m+1$ jobs of each block ``in parallel'', each at a rate of $\frac{m}{m+1}$
and finish the schedule after $k \cdot \frac{m+1}{m}$ time units in which each machine has always
been fully busy. This results in an integrality gap of at least $2 - \frac{2}{m+1}$.

It has been long known, that in principle one can take the linear program for any optimization problem and 
strengthen it automatically by applying an LP or SDP hierarchy lift.
We will provide formal definitions later, but basically these operators ensure that for any
set of at most $r$ variables, the LP solution indeed lies in the convex hull of integral combinations.
Here, $r$ is the number of \emph{levels} or \emph{rounds} and one typically needs
time $n^{O(r)}$ to solve an $r$-level hierarchy.

Some known approximation results have been reinterpreted in hindsight in this framework, for 
example a constant number of Lasserre rounds applied to a basic LP suffices for the 
Goemans-Williamson algorithm for MaxCut~\cite{MaxCut-GoemansWilliamson-JACM95} and also a constant 
number of Lasserre rounds
implies the triangle inequalities in the $O(\sqrt{\log n})$-approximation algorithm by 
Arora, Rao and Vazirani~\cite{SparsestCut-AroraRaoVazirani-JACM09}. Moreover, the
subspace enumeration component in the subexponential time algorithm of Arora, Barak and
Steurer~\cite{UniqueGamesAlgo-AroraBarakSteurer-FOCS10} for Unique Games could be replaced with a Lasserre SDP.
However, there are relatively few results where hierarchies
have been genuinely useful (at least fewer than researchers have hoped for). 
For example Chlamt{\'a}{\v c}~\cite{ApproxAlgoViaSDP-Chlamtac-FOCS07}
used SDP hierarchies to find better colorings in 3-colorable graphs and Raghavendra and 
Tan~\cite{CSPs-with-card-constraints-RaghavendraTanSODA12} apply them to obtain approximation
algorithms for CSPs with cardinality constraints. An application to color
hypergraphs can be found in~\cite{HypergraphColoringChlamtacSinghAPPROX08}.
Hierarchies also turned out to be the
right approach for Sparsest Cut in bounded tree width graphs, see the paper by 
Chlamt{\'a}{\v c}, Krauthgamer and Raghavendra~\cite{SparsestCutInBoundedTreeWidthGraphs-CKR-APPROX2010} 
and the 2-approximation by Gupta, Talwar and Witmer~\cite{SparsestCut-GuptaTalwarWitmer-STOC2013}.
For an application of the Lasserre hierarchy in the context of scheduling, see the recent work 
of Bansal, Srinivasan and Svensson~\cite{LiftAndRound-BansalSrinivasanSvensson-STOC2016}.
Throughout this paper, logarithms will be with respect to base 2, that means $\log(T) := \log_2(T)$.



\subsection{Our Contribution}

Our main result is that an LP lift with 
\[
(\log(n))^{O((m^2/\varepsilon^2) \cdot \log \log n)}
\]
rounds closes the integrality gap of LP~\eqref{eq:LP} to at most $1 + \varepsilon$. This implies:
\begin{theorem}
For the problem $Pm \mid \textrm{prec}, p_j = 1 \mid C_{\max}$ one can compute a $(1 + \varepsilon)$-approximate
solution in time $n^{O(r)}$ where $r := (\log(n))^{O((m^2/\varepsilon^2) \cdot \log \log n)}$. 
\end{theorem}
This gives a partial answer to one of the questions under ``Open Problem 1'' in \cite{TenOpenProblems-SchuurmanWoeginger1991}
which asked whether there is a PTAS for this problem. 
In a Dagstuhl workshop, Mathieu~\cite{DagstuhlOpenProblems2010} asked the more specific 
question whether the Sherali-Adams hierarchy gives a $(1 + \varepsilon)$-approximation 
after $c(\varepsilon,m)$ rounds. We also make progress on the question from the book of Garey and 
Johnson~\cite{GareyJohnson79}
by improving the $\frac{4}{3}$-polynomial time approximation for $m=3$~\cite{TwoSchedulingAlgorithms-Lam-Sethi-SICOMP77}
to a $1+\varepsilon$ in slightly more than quasi-polynomial time. In particular, this implies that  $Pm \mid \textrm{prec}, p_j = 1 \mid C_{\max}$  is
not $\mathbf{APX}$-hard, assuming that $\mathbf{NP} \not\subseteq \mathbf{DTIME}(n^{\log(n)^{O(\log\log n)}})$.




\section{An Explicit LP Hierarchy for Makespan Scheduling}

In principle, our result can be obtained by applying the well-known Sherali-Adams
hierarchy to the linear program in \eqref{eq:LP} --- of course the same still
holds true for even more powerful hierarchies such as the Lasserre SDP hierarchy. 
While this may be the preferable option for experts, we will work with 
an explicit strengthening of the above linear program that hopefully will be more
accessible to non-experts in LP hierarchies. For a set $K\subseteq \setR^m$ we denote
\[
\textrm{cone}(K) := \Big\{ \sum_{i=1}^k \lambda_ix_i \mid k \in \setN; \; x_i \in K \; \forall i \in [k]; \; \lambda_i \geq 0 \; \forall i \in [k] \Big\}
\] 
as the \emph{convex cone} that is spanned by $K$.

Let us fix a
parameter $r$.
Let $\sigma : J \to [T] \cup \{ *\}$ be a \emph{partial assignment} that
assigns slots
only for a subset of jobs. All the jobs with $\sigma(j) = *$ are unassigned.
Let $\textrm{supp}(\sigma) := \{ j \in J \mid j\textrm{ is assigned in
}\sigma\}$ be the support of that partial
assignment. We denote $\emptyset$ as the partial assignment that assigns
no job at all.
Moreover for a partial assignment $\sigma$ and $j \notin
\textrm{supp}(\sigma)$ and $t \in [T]$,
let $\sigma \cup (j,t)$ be the partial assignment augmented by
$\sigma(j) = t$.

We say that a solution to $\texttt{SA}(K(T),r)$ is a set of vectors 
$\bm{x} := \{x^\sigma\}_{|\textrm{supp}(\sigma)| \leq r}$,
where we define $x := x^\emptyset$
satisfying the following program:
\[
\begin{array}{rclll}
  x^{\sigma} &=& \sum_{t \in [T]} x^{\sigma \cup (j,t)} \hspace{2.3cm} \forall
\sigma: |\textrm{supp}(\sigma)| < r\textrm{ and }j \notin
\textrm{supp}(\sigma) & (I) \vspace{2mm} \\
  x^{\sigma} &\in& \textrm{cone}\big(K(T) \cap \{ x \mid x_{j,\sigma(j)} = 1 \;
\forall j \in \textrm{supp}(\sigma) \}\big) \quad \forall \sigma:
|\textrm{supp}(\sigma)| \leq r & (II) \vspace{2mm} \\
  x &\in& K(T) & (III)
\end{array}
\]
In other words, $\bm{x}$ is a collection of $n^{O(r)}$ many vectors $x^{\sigma}$
that each has dimension $|J| \cdot T$.
Note that if $x^{\sigma}$ is a non-zero vector, then it can be scaled to be a 
fractional solution in $K(T)$ that has all assignments of the partial assignment $\sigma$
integral. 
Notice that we have a variable for each $\sigma$ with $|\textrm{supp}(\sigma)| \leq r$, 
so one can find a feasible solution of the program in $n^{O(r)}$ time. 

One can think of this system as basically being the Sherali-Adams system, just that 
we do include more redundant variables that will make it easy to prove the needed
properties. 
First, we claim that if there exists a valid schedule $\sigma^*$, then
$\texttt{SA}(K(T),r) \neq \emptyset$. Here we can build a valid solution
by simply choosing $x^{\sigma}$ as the
characteristic vector of $\sigma^*$
if $\sigma$ and $\sigma^*$ agree. We set $x^{\sigma} = \bm{0}$ if there
is a job $j \in \textrm{supp}(\sigma)$ so that $\sigma(j) \neq \sigma^*(j)$.
We give the following useful properties:

\begin{lemma}\label{lem:HierarchyProperties}
Fix some $r$. Let $\bm{x} \in
\texttt{SA}(K(T),r)$. Let\footnote{Here $j \in J$ is any fixed job. But note that this definition does not depend on
the choice of $j$.} $\lambda_{\sigma} := \sum_{t=1}^T
x_{j,t}^{\sigma}$. Then the following holds
\begin{enumerate}
\item[a)]  If $\lambda_{\sigma}>0$, then
$\frac{x^{\sigma}}{\lambda_{\sigma}} \in K(T) \cap \{ x \mid
x_{j,\sigma(j)} = 1 \; \forall j \in \textrm{supp}(\sigma)\}$.
\item[b)] If $r=n$, then 
$x \in \textrm{conv}\big(K(T) \cap \{
0,1\}^{J \times [T]}\big)$.
\item[c)] Let $j^* \in J$ and $t^* \in [T]$ so that $\rho :=
x_{j^*,t^*}>0$.
Then taking $y^{\sigma} := \frac{1}{\rho} \cdot x^{\sigma \cup (j^*,t^*)}$ for each $\sigma$, one
has $\bm{y}= \{y^{\sigma}\}_{|\textrm{supp}(\sigma)| \leq r-1} \in \texttt{SA}(K(T),r-1)$ and $y_{j^*,t^*} = 1$.
Moreover,  $x_{j,t} = 0 \Rightarrow y_{j,t}= 0$ for all $j \in J$ and $t \in [T]$.
\end{enumerate}
\end{lemma}

\begin{proof}
We prove the following:
\begin{enumerate}
\item[a)] Follows from $(II)$ and the definition of $\lambda_{\sigma}$.
\item[b)] We can iteratively apply (I) to obtain 
\[
x =
\sum_{\sigma: J \to [T]} x^{\sigma} = \sum_{\sigma: J \to [T]:
\lambda_{\sigma} >0} \lambda_{\sigma} \cdot
\frac{x^{\sigma}}{\lambda_{\sigma}}.
\] 
By a),
$\frac{x^{\sigma}}{\lambda_{\sigma}}$ are $0/1$ vectors.
\item[c)] From the definition we can see that $(I),(II)$ are just
inherited. $(III)$ and $y_{j^*,t^*} = 1$
follow from the scaling. The implication $x_{j,t} = 0
\Rightarrow y_{j,t}= 0$
follows from $y_{j,t} = \frac{1}{\rho} \cdot
x^{(j^*,t^*)}_{j,t} \leq \frac{1}{\rho} x_{j,t}$.
\end{enumerate}
\end{proof}
If we have solution $\bm{x} \in \texttt{SA}(K(T),r)$
and variables $j^*, t^*$ with $x_{j^*,t^*}>0$, then \emph{conditioning on $x_{j^*,t^*}=1$} means to replace the solution $\bm{x} $ with the solution 
$\bm{y} = \{ y^{\sigma}\}_{|\textrm{supp}(\sigma)| \leq r-1} \in \texttt{SA}(K(T),r-1)$ described in Lemma~\ref{lem:HierarchyProperties}.c.

\section{An Overview}\label{sec:Overview}

In this section, we will give an overview over the different steps in our algorithm; the detailed 
implementation of some of the steps will be given in Section~\ref{sec:ReducingDependence}, Section~\ref{sec:SchedulingTopJobs} and Section~\ref{sec:AccountingDiscardedJobs}. 
For a given time horizon $T$, a \emph{feasible schedule} is an assignment $\sigma : J \to \{ 1,\ldots,T\}$
with $|\sigma^{-1}(t)| \leq m$ for all $t \in [T]$ and for all $j,j' \in J$ one has $j \prec j' \Rightarrow \sigma(j) < \sigma(j')$.
Formally, our main technical theorem is as follows: 
\begin{theorem} \label{thm:MainTechnicalTheorem}
For any 
solution 
$\bm{x} \in \texttt{SA}(K(T),r)$ with 
$r :=(\log n) ^{O((m^2/\varepsilon^2) \cdot \log\log n)}$,
one can find a feasible schedule $\sigma : J \to \setN$ of the jobs in 
time $n^{O(r)}$ so that 
\[
\max_{j \in J} \sigma(j) \leq (1 + \varepsilon) \cdot T.
\]
\end{theorem}
To obtain a $(1+\varepsilon)$-approximation, we can find the minimum value of $T$ so that $\texttt{SA}(K(T),r)\neq \emptyset$ 
with binary search and then compute a solution 
$\bm{x} \in \texttt{SA}(K(T),r)$. In particular, by virtue of being a relaxation, that value of $T$ will satisfy $T \leq OPT$, where 
$OPT$ is the makespan of the optimum schedule. For the sake of a simpler notation, we will assume that $T$
is a power of 2 --- if $2^{z-1} < T \leq 2^z$ for some integer $z$, then one can add $m \cdot (2^z-T)$ many dummy
jobs that all depend on each original job so that the algorithm will schedule the dummy jobs  at the
very end. Moreover we will assume that $\frac{1}{\varepsilon},m \leq \log(n)$ as otherwise the bound is meaningless.

The main routine of our algorithm will schedule jobs only within the time horizon $T$ of the LP-hierarchy solution, but we will allow it to \emph{discard} jobs. Formally this
means, we will find an assignment $\sigma : J \setminus J_{\textrm{discarded}} \to [T]$ that will not have assigned slots to jobs 
in $J_{\textrm{discarded}}$. Such an assignment will still be called ``feasible'' if apart from the load bound, 
the condition $j \prec j' \Rightarrow \sigma(j) < \sigma(j')$ is satisfied for all $j,j' \in J \setminus J_{\textrm{discarded}}$. 
In particular dependencies with discarded jobs play no role in this definition. 

The reason for this definition is that one can easily insert the discarded jobs at the very end of the algorithm: 
\begin{lemma} \label{lem:InsertingDiscardedJobs}
Any feasible schedule 
\[
\sigma : J \setminus J_{\textrm{discarded}} \to \{1,\ldots,T\}
\] 
can be modified in polynomial time 
to a feasible schedule 
\[
\sigma^* : J \to \{ 1,\ldots,T + |J_{\textrm{discarded}}|\}
\]
which also includes the previously discarded jobs. 
\end{lemma}
\begin{proof} 
Select any job $j^* \in J_{\textrm{discarded}}$. Since $\sigma$ is a valid schedule which respects all precedence
constraints in $J \setminus J_{\textrm{discarded}}$, there must be a time $t^*$ so that all jobs $j \prec j^*$ have $\sigma(j) \leq t^*$ and 
all jobs $j$ with $j^* \prec j$ have $\sigma(j) > t^*$. Then we insert an extra time unit after time $t^*$; in
this extra time slot, we only process $j^*$. We continue the procedure with inserting the next job from $J_{\textrm{discarded}} \setminus \{j^*\}$.
\end{proof}

Now, let us introduce some notation:
We can imagine the precedence order ``$\prec$'' as a directed transitive graph $G = (J,E)$ with the nodes
as jobs and edges $(j,j') \in E \Leftrightarrow j \prec j'$. In that view, let $\delta^+(j) := \{ j' \in J \mid j \prec j'\}$ be the jobs 
depending on $j$ and let  $\delta^-(j) := \{ j' \in J \mid j' \prec j\}$ be the jobs on which $j$ depends. 
Note that $\delta^+(j)$ and $\delta^-(j)$ are always distinct. We abbreviate $\delta(j) := \delta^+(j) \cup \delta^-(j)$ 
as the jobs that have any dependency with $j$. Finally, for a subset of jobs $J' \subseteq J$, 
let $\Delta(J') := \max\{ |\delta(j) \cap J'| + 1 \mid j \in J' \}$ be the 
maximum degree of a node in the subgraph induced by $J'$, counting also the node itself. 

We partition the time horizon $[T]$ into a balanced binary family $\pazocal{I}$ of intervals of lengths $T, \frac{T}{2}, \frac{T}{2^2},\ldots,2,1$. 
Let $\pazocal{I} := \pazocal{I}_0 \dot{\cup} \ldots \dot{\cup}
\pazocal{I}_{\log(T)}$ be the \emph{binary laminar family of intervals} that
we obtain by repeatedly partitioning
intervals into two equally-sized subintervals.
Recall that each \emph{level} $\pazocal{I}_{\ell}$ contains $2^{\ell}$
many interval $I \in \pazocal{I}_{\ell}$; each one consisting
of $|I| = \frac{T}{2^{\ell}}$ many time units. 
For each job $j \in J$ and each interval $I$, we now define 
$x_{j,I} := \sum_{t \in I} x_{j,t}$, which denotes how much 
of job $j$ will be scheduled somewhere within 
that interval $I$.

Our algorithm will schedule the jobs in  a recursive manner. The main claim is that for any interval $I^*$,  LP-hierarchy solution $\bm{x}^*$ and a set of jobs $J^*$ with $x_{j,I^*}^*=1$
we can schedule almost all jobs from $J^*$ within $I^*$ while respecting all precedence constraints. 

We use parameters $k := \frac{c_1m}{\varepsilon}\log \log (T)$ where $c_1>0$ is a large 
enough constant that we will choose 
in Section~\ref{sec:AccountingDiscardedJobs}, and $\delta := \frac{\varepsilon}{8k^2m2^{2k^2}\log(T)}$.
To get some intuition for the parameters, considering $\varepsilon$ and $m$ as fixed constants, 
one would have $k = \Theta(\log \log n)$ and $\delta = 1/\log(n)^{\Theta(\log \log n)}$.
Formally, the main technical lemma is the following: 
\begin{lemma} \label{lem:MainRecursiveAlgorithm}
Fix $\varepsilon>0$. Let $I^* \in \pazocal{I}$ be an interval from the balanced family of length $T^* := |I^*|$.
Let $\bm{x}^* \in \texttt{SA}(K(T),r^*)$ be an LP-hierarchy solution with 
\[
  r^* \geq \log(T^*) \cdot 2mk^2 \cdot 2^{k^2} / \delta.
\]
Let
$J^* \subseteq \{ j \in J : x_{j,I^*}^{*} = 1\}$.
Then one can find a feasible assignment $\sigma : J^* \setminus J_{\textrm{discarded}}^* \to I^*$ 
that discards only
\[
 |J_{\textrm{discarded}}^*| \leq \frac{\varepsilon}{2} \cdot \frac{\log (T^*)}{\log(T)} \cdot T^* + \frac{\varepsilon}{2m} \cdot |J^*|
\]
many jobs.
\end{lemma}

Before we move on to explain the procedure behind Lemma~\ref{lem:MainRecursiveAlgorithm}, we 
want to argue that it implies our main result, Theorem~\ref{thm:MainTechnicalTheorem}: 
\begin{proof}
We set $I^* := \{ 1,\ldots,T\}$ and $\bm{x}^{*} := \bm{x}$, then $J^* := J$ is a valid choice as trivially $x_{j,\{1,\ldots,T\}} = 1$
for any job. To satisfy the requirement of Lemma~\ref{lem:MainRecursiveAlgorithm} we need
\[
 \log(T) \cdot \frac{2mk^2 \cdot 2^{k^2}}{\delta} \leq (\log(n))^{O((\frac{m}{\varepsilon})^2 \log\log(n))}
\]
many levels of the hierarchy. Here we use that $k = \Theta(\frac{m}{\varepsilon} \log \log T)$, 
hence $2^{k} = (\log(T))^{\Theta(m/\varepsilon)}$ and $2^{k^2} = (2^k)^k = (\log(T))^{\Theta((\frac{m}{\varepsilon})^2 \log \log T)}$
(we want to point out that many of the lower order terms are absorbed into the $O$-notation
of the exponent and we assume that $\frac{1}{\varepsilon},m \leq \log(n)$).
Then Lemma~\ref{lem:MainRecursiveAlgorithm} returns a valid assignment $\sigma : J \setminus J_{\textrm{discarded}} \to [T]$ 
that discards only
\[
 |J_{\textrm{discarded}}| \leq \frac{\varepsilon}{2} \cdot \frac{\log(T)}{\log(T)} \cdot T + \frac{\varepsilon}{2m} \cdot |J|
 \leq \varepsilon \cdot T
\]
many jobs. Inserting those discarded jobs via Lemma~\ref{lem:InsertingDiscardedJobs} then
results in a feasible schedule of makespan at most  $(1+\varepsilon) \cdot T$.
\end{proof}

The rest of the manuscript will be devoted to proving Lemma~\ref{lem:MainRecursiveAlgorithm}. 
We fix a constant $\varepsilon > 0$ as the target value for our approximation ratio
and denote $T^* := |I^*|$ as the length of our interval.

Let us first argue how to handle the base case, which for us is if $\log(T^*) \leq k^2$. 
In that case, we have at most $mT^* \leq m2^{k^2}$ jobs. Hence, the 
LP-hierarchy lift has  $r^* \geq mT^*$ many levels and one can repeatedly condition on events
$x_{j,t}=1$ for $j \in J^*$ and $t \in I^*$ until one arrives at 
an LP hierarchy solution $\bm{x}^{**}$
with $x_{j,t}^{**} \in \{ 0,1\}$ for all $j \in J^*$. This then represents a valid schedule of
jobs $J^*$ in the interval $I^*$ without the need to discard any jobs.

We now come to a high-level description of the algorithm. 
Let $\pazocal{I}_0^*,\ldots,\pazocal{I}_{\log(T^*)}^*$ be the family of
subintervals of $I^*$, where $\pazocal{I}_{\ell}^*$ contains $2^{\ell}$ intervals of length $\frac{T^*}{2^{\ell}}$
each, see Figure~\ref{fig:DissectionForIStar}. 
For a job $j \in J^*$, we define $\ell(j,\bm{x}^*) := \max\{ \ell : \exists I \in \pazocal{I}_{\ell}^*\textrm{ with } \sum_{t \in I}x_{j,t}^{*} = 1\}$ as the level that \emph{owns} the job in the current 
LP-hierarchy solution. 
We also abbreviate $J(\ell,\bm{x}^*) := \{ j \in J^{*} \mid \ell(j,\bm{x}^*) = \ell\}$ as 
all jobs owned by level $\ell$. The algorithm is as follows:
\begin{itemize}
\item {\bf Step 1:} Starting with the 
LP-hierarchy solution $\bm{x}^*$, we can iteratively condition on events 
until we arrive at 
a solution $\bm{x}^{**}$
that has the property that for any interval $I \in \pazocal{I}_0^* \cup \ldots \cup \pazocal{I}_{k^2-1}^*$, the jobs owned by that interval have small dependence degree, that means $\Delta(J(I,\bm{x}^{**})) \leq \delta |I|$, where $J(I,\bm{x}^{**}) := \{ j \in J^* \mid I\textrm{ minimal with }\sum_{t \in I} x^{**}_{j,t}=1\}$.
If we then consider the set of jobs $J^{**} := \{ j \in J^* \mid 0 \leq \ell(j,\bm{x}^{**}) < k^2 \}$ 
owned by the first $k^2$ levels, 
the longest chain in $J^{**}$ will contain at most $k^2\delta T^*$ jobs.
We will show in Section~\ref{sec:ReducingDependence} that the number of required conditionings can be upperbounded by $2mk^2 \cdot 2^{k^2} / \delta$, 
which implies that $\bm{x}^{**} \in \texttt{SA}(K(T),r^* - 2mk^2 \cdot 2^{k^2} / \delta)$. 
\item {\bf Step 2:} From now on, we work with the modified 
LP-hierarchy solution $\bm{x}^{**}$. We select a 
level index $\ell^* \in \{ k,\ldots,k^2\}$ and partition the jobs in $J^*$ in three different groups: 
  \begin{itemize}
  \item The jobs on the top levels: $J_{\textrm{top}} := J(0,\bm{x}^{**}) \cup \ldots \cup J(\ell^*-k-1,\bm{x}^{**})$
  \item The jobs on the $k$ middle levels: $J_{\textrm{middle}} := J(\ell^*-k,\bm{x}^{**}) \cup \ldots \cup J(\ell^*-1,\bm{x}^{**})$
  \item The jobs on the bottom levels: $J_{\textrm{bottom}} := J(\ell^*,\bm{x}^{**}) \cup \ldots \cup J(\log(T^*),\bm{x}^{**})$
  \end{itemize}
Then we discard all jobs in $J_{\textrm{middle}}$. In Section~\ref{sec:AccountingDiscardedJobs} we will describe how the
index $\ell^*$ is chosen and in particular we will provide an upper bound on the number of 
discarded middle jobs. 
\item {\bf Step 3:} In this step, we will find a schedule for the bottom jobs. 
For this purpose, we call Lemma~\ref{lem:MainRecursiveAlgorithm} \emph{recursively} for each interval $I \in \pazocal{I}_{\ell^*}$ 
with a copy of the 
solution $\bm{x}^{**}$ and jobs $J_I := \{ j \in J_{\textrm{bottom}} \mid x_{j,I}^{**} = 1\}$. Here it is crucial that the intervals are disjoint but also 
the sets $J_I$ are disjoint for different intervals $I \in \pazocal{I}_{\ell^*}$.
Then Lemma~\ref{lem:MainRecursiveAlgorithm} returns a valid schedule of the form $\sigma_I : J_I \setminus J_{I,\textrm{discarded}} \to I$ for each interval $I \in \pazocal{I}_{\ell^*}$.
Let $J_{\textrm{bottom-discarded}} := \bigcup_{I \in \pazocal{I}_{\ell^*}} J_{I,\textrm{discarded}} \subseteq J_{\textrm{bottom}}$ be the union of jobs that were discarded
in those calls. The partial schedules $\sigma_I$ satisfy $|\sigma_I^{-1}(t)| \leq m$ for $t \in I$ and $|\sigma_I^{-1}(t)| = 0$
for $t \notin I$. We combine those schedules to a schedule  
\[
\sigma : J_{\textrm{bottom}} / J_{\textrm{bottom-discarded}} \to I^*.
\] 
From the disjointness of the intervals, it is clear that again  $|\sigma^{-1}(t)| \leq m$ for all $t \in I^*$. 
Moreover, if
$j \prec j'$ and $j,j' \in J_I$ for some interval $I \in \pazocal{I}_{\ell^*}$, then by the inductive hypothesis $\sigma(j) < \sigma(j')$. 
On the other hand, if $j \in J_I$ and $j' \in J_{I'}$ then we know by Lemma~\ref{lem:HierarchyProperties}.c that $I$ had to come before $I'$ since $x_{j,I}^{**} = 1 = x_{j',I'}^{**}$. 
\item {\bf Step 4:} We continue working with the previously constructed schedule   $\sigma$ that schedules 
the non-discarded bottom jobs. In this step, we will extend the schedule $\sigma$ and insert the jobs 
of $J_{\textrm{top}}$ in the remaining free slots. We will prove in Section~\ref{sec:SchedulingTopJobs} that this can be done 
without changing the position of any scheduled bottom job and without violating any precedence constraints. 
Again, we allow that the procedure discards a small number of additional jobs 
from $J_{\textrm{top}}$ that we will account for later. Eventually, the schedule $\sigma$ satisfies the claim 
for Lemma~\ref{lem:MainRecursiveAlgorithm}.
\end{itemize}

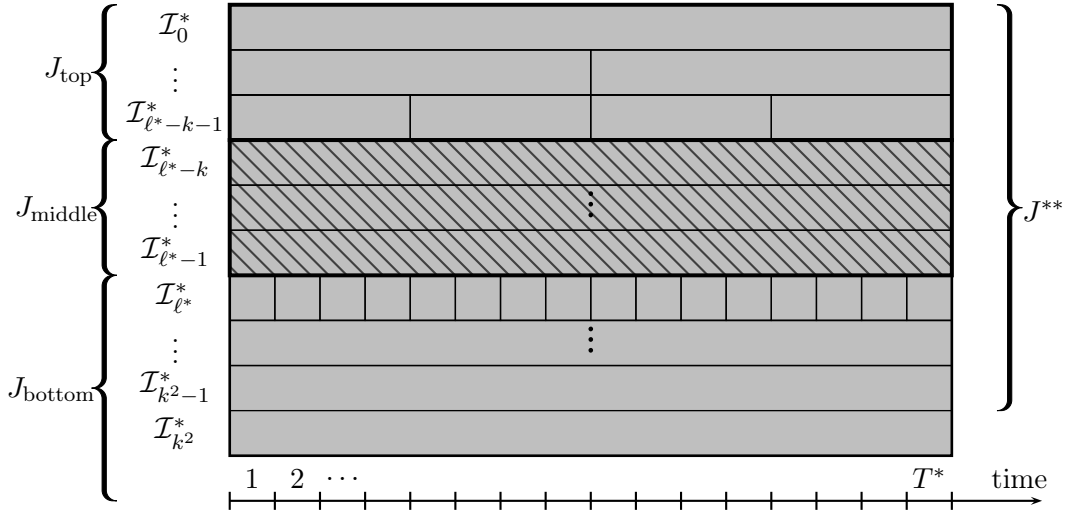
\begin{figure}
\begin{center}
\psset{xunit=0.6cm,yunit=0.6cm}
\begin{pspicture}(-2,-2)(15,9.3)
\drawRect{linewidth=1.0pt,fillstyle=solid,fillcolor=lightgray}{0}{-1}{16}{10}
\drawRect{linewidth=1.5pt,fillstyle=solid,fillcolor=lightgray}{0}{6}{16}{3}
\drawRect{linewidth=1.5pt,fillstyle=vlines,fillcolor=lightgray,hatchcolor=darkgray}{0}{3}{16}{3}
\drawRect{linewidth=0.5pt,fillstyle=none,fillcolor=lightgray}{0}{7}{8}{1}
\drawRect{linewidth=0.5pt,fillstyle=none,fillcolor=lightgray}{8}{7}{8}{1}
\multido{\N=0+4}{4}{\drawRect{linewidth=0.5pt,fillstyle=none,fillcolor=lightgray}{\N}{6}{4}{1}}
\multido{\N=0+1}{16}{\drawRect{linewidth=0.5pt,linestyle=solid,fillstyle=none,fillcolor=lightgray}{\N}{2}{1}{1}}
\multido{\N=0+1}{9}{\drawRect{linewidth=0.5pt,fillstyle=none,fillcolor=lightgray}{0}{\N}{16}{1}}
\rput[c](8,4.75){$\Huge{\vdots}$}
\rput[c](-1.2,8.5){$\pazocal{I}_0^*$}
\rput[c](-1.2,7.5){$\vdots$}
\rput[c](8,1.75){$\Huge{\vdots}$}
\rput[c](-1.2,6.5){$\pazocal{I}_{\ell^*-k-1}^*$}
\rput[c](-1.2,5.5){$\pazocal{I}_{\ell^*-k}^*$}
\rput[c](-1.2,4.5){$\vdots$}
\rput[c](-1.2,3.5){$\pazocal{I}_{\ell^*-1}^*$}
\rput[c](-1.2,2.5){$\pazocal{I}_{\ell^*}^*$}
\rput[c](-1.2,1.5){$\vdots$}
\rput[c](-1.2,0.5){$\pazocal{I}_{k^2-1}^*$}
\rput[c](-1.2,-0.5){$\pazocal{I}_{k^2}^*$}
\psbrace[nodesepB=5pt,braceWidthInner=4pt,braceWidthOuter=4pt](17,0)(17,9){$J^{**}$}
\psline{->}(0,-2)(18,-2) \multido{\N=0+1}{17}{\psline(\N,-1.8)(\N,-2.2)}
\rput[c](0.5,-1.5){$1$}
\rput[c](1.5,-1.5){$2$}
\rput[c](2.5,-1.5){$\ldots$}
\rput[c](15.5,-1.5){$T^*$}
\rput[c](17.5,-1.5){time}
\psbrace[nodesepB=5pt,braceWidthInner=4pt,braceWidthOuter=4pt,rot=180,nodesepA=-18pt](-2.5,9)(-2.5,6){$J_{\textrm{top}}$}
\psbrace[nodesepB=5pt,braceWidthInner=4pt,braceWidthOuter=4pt,rot=180,nodesepA=-29pt](-2.5,6)(-2.5,3){$J_{\textrm{middle}}$}
\psbrace[nodesepB=5pt,braceWidthInner=4pt,braceWidthOuter=4pt,rot=180,nodesepA=-32pt](-2.5,3)(-2.5,-2){$J_{\textrm{bottom}}$}
\end{pspicture}
\end{center}
\caption{Binary dissection of the interval $I^*$ used in the algorithm behind Lemma~\ref{lem:MainRecursiveAlgorithm}.\label{fig:DissectionForIStar}}
\end{figure}

The intuition behind the algorithm is as follows:  When we call the procedure recursively for intervals $I \in \pazocal{I}_{\ell^*}^*$
we cannot control where the jobs $J_I$ will be scheduled within that interval $I$. In particular the
decisions made in different intervals $I,I' \in \pazocal{I}_{\ell^*}^*$ will in general not be consistent.
But the discarding of the middle jobs creates a gap between the top 
jobs and the bottom jobs in the sense that the intervals of the top jobs are at least a factor $2^k$
longer than intervals of the bottom jobs. For a top job $j \in J_{\textrm{top}}$ we will be pessimistic and
assume that all the bottom jobs that $j$ depends on will be scheduled just at the very end of their
interval. Still, as those intervals are very short, we will be able to argue that the loss in the
flexibility is limited and most of the top jobs can be processed. As a second crucial ingredient, the
conditioning had the implication that the top jobs do not contain any long chains any more.
This will imply that a greedy
schedule of the top jobs will leave little idle time, resulting in only few discarded top jobs.

A high-level pseudo-code description of the whole scheduling algorithm can be found in Figure~\ref{fig:MainAlgorithm}:
\begin{figure}
\begin{center}
\psframebox{
\begin{minipage}{14.3cm}
$\textsc{QPTAS for Makespan Scheduling}$ \vspace{1mm} \hrule  \vspace{1mm}
{\bf Input:} Scheduling instance $(J,\prec)$, parameters $m \in \setN$ and $\varepsilon>0$ \\
{\bf Output:} $(1+\varepsilon)$-approximate schedule $\sigma$ \vspace{1mm} \hrule  \vspace{1mm}
\begin{enumerate*}
\item[(1)] Compute a solution $\bm{x} = (x^{\sigma})_{|\textrm{supp}(\sigma)| \leq r} \in \texttt{SA}(K(T),r)$ with $r := (\log(n))^{O(m^2/\varepsilon^2) \cdot \log \log(n)}$ and $T$ minimal
\item[(2)] Call $\textsc{RecursiveScheduling}(J,\bm{x},[T]) \to \sigma$
\item[(3)] Insert discarded jobs into schedule $\sigma$
\end{enumerate*}
\vspace{1mm} \hrule \vspace{1mm} \hrule  \vspace{1mm}
$\textsc{RecursiveScheduling}$ \vspace{1mm} \hrule  \vspace{1mm}
{\bf Input:} Jobs $J^*$, LP lift $\bm{x}^*$, interval $I^*$ with $\sum_{t \in I^*}x_{j,t}^*=1$ for $j \in J^*$ \\
{\bf Output:} Schedule $\sigma$ with some jobs discarded \vspace{1mm} \hrule  \vspace{-2mm}
\begin{enumerate*}
\item[(1)] Build binary family of intervals $\pazocal{I}^*$
\item[(2)] Call ${\textsc{Breaking Long Chains}}$ $\to \bm{x}^{**}$
\item[(3)] Select partition into top, middle, bottom jobs. Pick $\ell^*$. 
\item[(4)] Discard middle jobs.
\item[(5)] For each $I \in \pazocal{I}_{\ell^*}$ set $J_I := \{ j \in J^* \mid x_{j,I}^{**}=1\}$ and \\ call $\textsc{RecursiveSchedule}(J_I,\bm{x}^{**},I) \to \sigma_I$
\item[(6)] Combine $\sigma_I$'s to one schedule $\sigma$
\item[(7)] Call two-phased algorithm based on matching and EDF to insert top jobs into $\sigma$
\end{enumerate*}
\end{minipage}}
\end{center}
\caption{High-level description of main algorithm\label{fig:MainAlgorithm}.}
\end{figure}

\section{Step (1) --- Reducing Dependence} \label{sec:ReducingDependence}

In this section we will implement ``Step (1)'' and show that we can reduce the maximum dependence degrees of 
the jobs owned by the first $k^2$ levels in order to bound the length of chains. 
We will do this by conditioning on up to
$2mk^2 \cdot 2^{k^2} / \delta$ many variables.
We are considering an interval $I^*$ and 
a subset of jobs $J^*  \subseteq J$ that the vector $x^{*}$ from the current 
LP-hierarchy solution $\bm{x}^*$ fully
schedules within $I^*$. Recall that for one of the subintervals $I \in \pazocal{I}_{\ell}^*$ 
below $I^*$, we write $J(I,\bm{x}^*) = \{ j \in J(\ell,\bm{x}^{*}) \mid x_{j,I}^{*} = 1 \}$
as the jobs owned by that particular interval.
\begin{lemma} \label{lem:BreakingChains}
Let 
$\bm{x}^* \in \texttt{SA}(K(T),r^*)$.
Then one can find an induced solution $\bm{x}^{**} \in \texttt{SA}(K(T),r^{**})$
with $r^{**} := r^* -2mk^2 \cdot 2^{k^2} / \delta$
so that $\Delta(J(I,\bm{x}^{**})) \leq \delta \cdot |I|$ for all
intervals $I \in \pazocal{I}_0^* \cup \ldots \cup \pazocal{I}_{k^2-1}^*$.
\end{lemma}
\begin{proof}
We set initially $\bm{x}^{**}:= \bm{x}^*$. If there is any interval $I = I_1
\dot{\cup} I_2 \in \pazocal{I}_0 \cup \ldots \cup \pazocal{I}_{k^2-1}$ with
$\Delta(J(I,\bm{x}^{**})) > \delta \cdot |I|$, then we must have a job
$j \in J(I,\bm{x}^{**})$
that has either $|\delta^+_{J(I,\bm{x}^{**})}(j) \cup \{ j\}| \geq
\frac{\delta}{2} \cdot  |I|$ or 
$|\delta^-_{J(I,\bm{x}^{**})}(j) \cup \{ j\}| \geq
\frac{\delta}{2} \cdot  |I|$. We assume that 
$|\delta^+_{J(I,\bm{x}^{**})}(j) \cup \{ j\}| \geq
\frac{\delta}{2} \cdot  |I|$ holds and omit the other case, which is symmetric.
Then we pick a time $t \in I_2$ with $x^{**}_{j,t} > 0$ and
 replace $\bm{x}^{**}$ by the 
LP-hierarchy solution conditioned on the
event ``$x_{j,t}^{**} = 1$''. Note that this
means that all jobs in $\delta_{J(I,\bm{x}^{**})}^+(j) \cup \{j\}$ will be removed
from $J(I,\bm{x}^{**})$. In fact, each such job will be moved 
to $J(I',\bm{x}^{**})$ where $I' \subseteq I_2$
is some subinterval. 
The conditioning can also change the owning interval of
other jobs, but for each job $j$, the set  
of times $t$ such that $x_{j,t}^{**}>0$
can only \emph{shrink} if we condition on
any event, see Lemma~\ref{lem:HierarchyProperties}.c. 
Hence jobs only move from intervals to subintervals. 

Since in each iteration, at least $\frac{\delta}{2} \cdot |I| \geq \frac{\delta}{2} \cdot  \frac{T^*}{2^{k^2}}$ many
jobs ``move'' and each job moves at most 
$k^2$ many
times out of an interval
in $\pazocal{I}_0^* \cup \ldots \cup \pazocal{I}_{k^2-1}^*$, we need to condition at most 
\[
\frac{2mT^* \cdot k^2}{\delta \frac{T^*}{2^{k^2}}} = 2mk^2 \cdot \frac{2^{k^2}}{\delta}
\]
many times.
\end{proof}

The implication of Lemma~\ref{lem:BreakingChains} is that the set of jobs owned by
intervals $I \in \pazocal{I}_0^* \cup \ldots \cup \pazocal{I}_{k^2-1}^*$ will not contain long chains,
simply because we have only few intervals and none of jobs owned
by a single interval contain long chains anymore.
\begin{lemma} \label{lem:BoundedChains}
After applying Lemma~\ref{lem:BreakingChains}, the longest chain within
jobs owned by intervals $I \in \pazocal{I}_0^* \cup \ldots \cup \pazocal{I}_{k^2-1}^*$
has length at most $k^2 \delta T^*$.
\end{lemma}
\begin{proof}
First, let us argue how many jobs a chain can have that are all assigned
to intervals of the same level
$\ell$. From each interval $I$, the chain can only include $\delta |I| = \delta \cdot
\frac{T^*}{2^{\ell}}$ many jobs.
Since $|\pazocal{I}_{\ell}| = 2^{\ell}$, the total number of jobs from
level $\ell$ is bounded by $\delta T^*$.
The claim follows from the pigeonhole principle and the fact that we have $k^2$
many levels in $\pazocal{I}_0^* \cup \ldots \cup \pazocal{I}_{k^2-1}^*$.
\end{proof}

We can summarize the algorithm from Lemma~\ref{lem:BreakingChains} as follows: 
\begin{center}
\psframebox{
\begin{minipage}{14.1cm}
{\sc Breaking long chains} \vspace{1mm} \hrule  \vspace{1mm}
{\bf Input:} Scheduling instance with jobs $J^*$, a precedence order, an 
LP-hierarchy solution $\bm{x}^* \in \texttt{SA}(K(T),r^*)$, an interval $I^*$ \\
{\bf Output:} An 
LP-hierarchy solution $\bm{x}^{**}$ with maximum chain length $k^2\delta T^*$ in $\bigcup_{\ell=0}^{k^2-1} J(\ell,\bm{x}^{**})$   \vspace{1mm} \hrule \vspace{2mm}
\begin{enumerate*}
\item[(1)] Make a copy $\bm{x}^{**} := \bm{x}^*$
   \item[(2)] WHILE $\exists I = (I_1\dot{\cup}I_2)\in \bigcup_{\ell=0}^{k^2-1} \pazocal{I}_{\ell}^*$ WITH  
$\Delta(J(I,\bm{x}^{**})) > \delta |I|$ DO 
      \begin{enumerate*}
      \item[(3)] Choose a job $j \in J(I,\bm{x}^{**})$ with $|\delta_{J(I,\bm{x}^{**})}(j) \cup \{j\}| \geq \delta |I|$. 
      \item[(4)] If $|\delta_{J(I,\bm{x}^{**})}^+(j) \cup \{ j\}| \geq \frac{\delta}{2} \cdot |I|$ THEN condition in $\bm{x}^{**}$ on $x_{j,t}^{**}=1$ for some $t\in I_2$\\
ELSE condition on $x_{j,t}^{**}=1$ for some $t\in I_1$. 
      \end{enumerate*} 
\end{enumerate*}
\end{minipage}}
\end{center}
Note that after each conditioning in step (6), the solution $\bm{x}^{**}$ will change and the 
set $J(I,\bm{x}^{**})$ will be updated. 


\section{Step (4) --- Scheduling Top Jobs}\label{sec:SchedulingTopJobs}

Consider the algorithm from Section~\ref{sec:Overview} and the state at the end of
Step 3. At this point, we have a schedule $\sigma$ that schedules most of the bottom jobs.
The main argument that remains to be shown is how to add in the top jobs
which are owned by intervals in $\pazocal{I}_0^* \cup \ldots \cup \pazocal{I}_{\ell^* - k - 1}^*$.
This is done in two steps.
First, we use a \emph{matching-based} argument to show that most top jobs can be inserted 
in the existing schedule so that
the precedence constraints with the bottom jobs are respected. In this step, we will be 
discarding up to  $4m \cdot 2^{-k} \cdot T^*$ many jobs. More crucially, the schedule
will not have satisfied precedence constraints within $J_{\textrm{top}}$.
In a 2nd step, we temporarily remove the top jobs from the schedule and reinsert them 
with a variant of the \emph{Earliest Deadline First (EDF)} scheduling. 
As we will see later in Theorem~\ref{thm:EDF}, this results in at most 
 $\frac{\varepsilon}{8\log T} \cdot T^*$ additionally discarded jobs. 

\subsection{A Preliminary Assignment of Top Jobs}

Let us recall what we did so far. In Step 3, we applied Lemma~\ref{lem:MainRecursiveAlgorithm} 
recursively on each interval $I \in \pazocal{I}_{\ell^*}$ to schedule the bottom jobs. We already
argued that the resulting schedules could be combined to a schedule 
$\sigma : J_{\textrm{bottom}} / J_{\textrm{bottom-discarded}} \to I^*$ that respects all precedence
constraints. 

Let the intervals in $\pazocal{I}_{\ell^*}^*$ be called $I_1,\ldots,I_p$, so that  the time horizon $T^*$ is partitioned into
$p$ equally sized
\emph{subintervals} with $p = 2^{\ell^*}$. After reindexing the time horizon, let us assume for the
sake of a simpler notation that $I^* = \{ 1,\ldots,T^*\}$. If we abbreviate $t_i := i \cdot
\frac{T^*}{p}$ for $i \in \{ 0,\ldots,p\}$,
then the $i$th interval contains the time periods $I_i := \{
t_{i-1}+1,\ldots,t_i\}$.
Each time $t$ has an available \emph{capacity} of $\textrm{cap}(t) = m-|\sigma^{-1}(t)| \in \{ 0,\ldots,m\}$ many machines, which is the number of machines not used by jobs in $J_{\textrm{bottom}}$.
We abbreviate $\textrm{cap}(I_i) := \sum_{t \in I_i} \textrm{cap}(t)$ as
the capacity of interval $I_i$.

The available positions of jobs in $J_{\textrm{top}}$ are constrained by the scheduled times of jobs in $J_{\textrm{bottom}}$. As we had no further control over the exact position of the bottom jobs within their intervals $I_i$, we want to define for each job $j \in J_{\textrm{top}}$  a \emph{release time} $r_j$ and a \emph{deadline} $d_j$ determined by the most pessimistic outcome of how $\sigma$ could have scheduled
the bottom jobs.  
%
For all $j \in J_{\textrm{top}}$, we define
\begin{eqnarray}
r_j &:=& \min\left\{t_i+1 \mid \sigma(j')\leq t_i \; \forall j'\in J_{\textrm{bottom}} : j' \prec j\right\} \label{eq:ReleaseTimeDeadlineDef}\\
d_j &:=& \max\left\{t_i \mid \sigma(j')\geq t_i+1 \; \forall j'\in J_{\textrm{bottom}} : j \prec j'\right\} \nonumber
\end{eqnarray}
In particular, the release time will be the first time unit of an interval $I_i$ and the deadline
will be the last time unit of an interval $I_i$. 
Let $i_r(j)$ and $i_{d}(j)$ be the corresponding indices, so that the release time
is of the form $r_j = t_{i_r(j)-1}+1$ and the deadline is $d_j = t_{i_d(j)}$. Then our goal is to show that
most top jobs $j$ can be scheduled somewhere in the time frame
$I_{i_r(j)} \cup \ldots \cup I_{i_d(j)}$. This would imply that at least all precedence constraints between bottom and top jobs are going to be satisfied. 

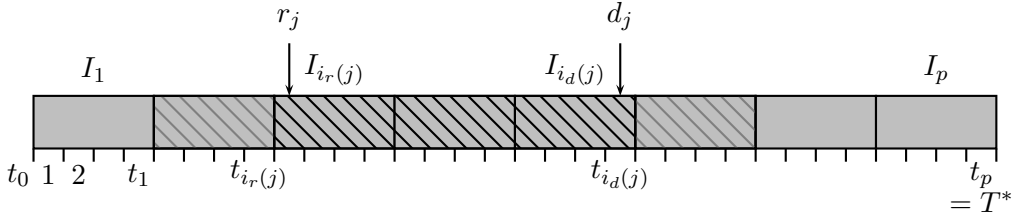
\begin{figure}
\begin{center}
\psset{xunit=0.80cm,yunit=0.7cm}
\begin{pspicture}(0,-0.8)(16,3.9)
\multido{\N=0+2}{8}{
\drawRect{fillstyle=solid,fillcolor=lightgray}{\N}{0}{2}{1}
}
\multido{\N=0.0+0.5}{33}{
\psline(\N,0)(\N,-5pt)
}
\drawRect{fillstyle=vlines}{4}{0}{6}{1}
\drawRect{fillstyle=vlines,hatchcolor=gray}{2}{0}{2}{1}
\drawRect{fillstyle=vlines,hatchcolor=gray}{10}{0}{2}{1}
\rput[c](-0.25,-10pt){$t_0$}
\rput[c](0.25,-10pt){$1$}
\rput[c](0.75,-10pt){$2$}
\rput[c](1.75,-10pt){$t_1$}
\rput[c](3.75,-10pt){$t_{i_r(j)}$}
\rput[c](9.75,-10pt){$t_{i_d(j)}$}
\rput[c](15.75,-10pt){$t_{p}$} \rput[c](15.75,-20pt){$=T^*$}
\rput[c](1,1.5){$I_1$}
\rput[c](5,1.5){$I_{i_r(j)}$}
\rput[c](9,1.5){$I_{i_d(j)}$}
\rput[c](15,1.5){$I_{p}$}
\pnode(4.25,1){A} \pnode(4.25,2){B} \ncline{->}{B}{A} \nput{90}{B}{$r_j$}
\pnode(9.75,1){A} \pnode(9.75,2){B} \ncline{->}{B}{A} \nput{90}{B}{$d_j$}
\end{pspicture}
\caption{Visualization of interval $I^* = I_1 \dot{\cup} \ldots \dot{\cup} I_{p}$ and possible release times and deadlines for a job $j \in J_{\textrm{top}}$. Note that 
$x^{**}$ might schedule $j$ over the whole hatched area, while our choice of $r_j$ and $d_j$ will force us to process $j$ inside the black-hatched area (or to discard the job).\label{fig:ReleaseTimeDeadline}}
\end{center}
\end{figure}

Notice here that the existing fractional assignment that $\bm{x}^{**}$ provides for $j \in J_{\textrm{top}}$ might 
also be using the slots in the two intervals coming right before and right after  the range $[r_j, d_j]$ (see Figure~\ref{fig:ReleaseTimeDeadline}). This is due to our rounding of release times and deadlines to interval beginnings and ends.
Let $J_{\textrm{bottom}}' := J_{\textrm{bottom}} \setminus J_{\textrm{bottom-discarded}}$ be the bottom jobs 
scheduled by the recursive calls of the algorithm.

\begin{lemma} \label{lem:MakeRoomForJobs} A valid schedule $\sigma : J_{\textrm{bottom}}' \to I^*$ of bottom jobs can be extended to a schedule $\sigma : J_{\textrm{bottom}}' \cup J_{\textrm{top}}' \to I^*$ 
with $J_{\textrm{top}}' \subseteq J_{\textrm{top}}$
that includes most of the top jobs. The schedule satisfies (i) $|\sigma^{-1}(t)| \leq m$ for $t \in I^*$; (ii)  $r_j \leq \sigma_j \leq d_j$ for all $j \in J_{\textrm{top}}'$ and (iii) one discards at most $|J_{\textrm{top}} \setminus J_{\textrm{top}}'| \leq 4m \cdot 2^{-k} \cdot T^*$ many top jobs. 
\end{lemma}

\begin{proof}
We want to use a matching-based argument.
For this sake, we consider the bipartite graph with jobs on one side and subintervals on the other. Formally, we define a graph $G = (V, U, E^+)$ with $V = J_{\textrm{top}}$, $U = \{1,\ldots,p\}$ where the nodes $i \in U$ have capacity $\textrm{cap}(I_i)$, and edges 
\[
\hspace{-3mm} E^+ = \{(j,i) \in V \times U  \mid \max\{i_{r}(j) - 1, 1\} \leq i \leq \min\{i_{d}(j)+1,p\}\}.
\]
We say that a matching $M$ is \emph{$V$-perfect} if it covers every node in $V$.
Then the neighborhood of each top job includes every interval in which it is fractionally scheduled in $\bm{x}^{**}$. Moreover, each of the bottom jobs $j \in J_{\textrm{bottom}}'$ has been 
assigned by $\sigma$ to precisely that interval $I_i$ with $x_{j,I_i}^{**}=1$. Hence $\bm{x}^{**}$ gives a
$V$-perfect fractional matching that respects the given capacities $\textrm{cap}(I_i)$.
In bipartite graphs, the existence of a \emph{fractional} $V$-perfect matching implies
the existence of an \emph{integral} $V$-perfect matching, see e.g.~\cite{CombinatorialOptimizationBook-Schrijver-2003}. 

However, in order to assign the top jobs to slots within release times and deadlines
we are only allowed to use the smaller set of edges
\[
E = \{(j,i) \in V \times U \mid i_{r(j)} \leq i \leq i_{d_j}\}.
\]
For any $J^* \subseteq V$, we let $N^+(J^*)$ be the neighborhood of $J^*$ along edges in $E^+$ and let $N(J^*)$ be its neighborhood along edges in $E$. Let the magnitude of a neighborhood $|N(J^*)|$ be defined as the sum of capacities of the nodes it contains.
By \emph{Hall's Theorem}~\cite{CombinatorialOptimizationBook-Schrijver-2003}, the minimum number of \emph{exposed} $V$-nodes in a maximum matching in $E$ is precisely
\[
  \max_{J \subseteq V} \{ |J| - |N(J)| \}.
\]
Now, fix the set $J^* \subseteq V$ attaining the maximum; then $|J^*| - |N(J^*)|$
is the number of jobs that we have to discard. Since $E^+$ allows for a 
$V$-perfect matching, the reverse direction of Hall's Theorem gives that
$|J^*| \leq |N^+(J^*)|$. Thus $|J^*| - |N(J^*)| \leq |N^+(J^*)| - |N(J^*)|$.
Note that $N(J^*)$ is in general not a consecutive interval of $\{1,\ldots,p\}$.
We can upper bound the difference $|N^+(J^*)| - |N(J^*)|$ by $2m \cdot \frac{T^*}{p}$
times the number of connected components of $N(J^*)$. 
This is the point where we take advantage of the ``gap'' between the levels of the top and bottom jobs: for each job $j \in V$ there is an interval $I \in \pazocal{I}_0^* \dot{\cup} \ldots \dot{\cup} \pazocal{I}_{\ell^* - k - 1}^*$ so that $N^+(j)$
contains the midpoint of that interval. Due to the gap, there are at most $p \cdot 2\cdot 2^{-k}$ such intervals\footnote{It is possible that $N(j)=\emptyset$. Still $N^+(j)$ will contain a midpoint of a level $0,\ldots,\ell^*-k-1$ interval, hence we have accounted for those jobs as well. Note that such jobs would automatically get discarded.}. Hence the number of discarded jobs can be bounded by
\[
  |J^*| - |N(J^*)| \leq 2m \cdot \frac{T^*}{p} \cdot 2p \cdot 2^{-k} = 4m \cdot 2^{-k} \cdot T^*. 
\]
Finally note that a corresponding maximum matching can be computed in polynomial time.
\end{proof}

\subsection{Reassigning the Top Jobs via EDF}\label{sec:ReassigningTopJobsViaEDF}

We have seen so far that we can schedule most of the bottom and top jobs so that 
all precedence constraints within the bottom jobs are satisfied and the
top jobs are correctly scheduled between the bottom jobs that they depend on. However, 
the schedule as it is now ignores the precedence constraints within the top jobs. 
In this section, we will remove the top jobs from the schedule and then reinsert them
using a variant of the \emph{Earliest Deadline First} (EDF) scheduling policy. 

For the remainder of Section~\ref{sec:ReassigningTopJobsViaEDF}, we will show a 
stand-alone theorem that we will use as a black box. Imagine a general setting where
we have $m$ identical machines and $n$ jobs $J$, each one with integer release 
times $r_j$ and deadlines $d_j$ and a unit processing time. The EDF scheduling rule picks
at any time the available job with minimal $d_j$ for processing.
It is a classical result in scheduling theory by Dertouzos~\cite{EDF-Dertouzos74} that for $m=1$
and unit size jobs, EDF is an \emph{optimal} policy. Here ``optimal'' means that if there
is any schedule that finishes all jobs before their deadline, then EDF does so, too.
The result extends to the case of arbitrary running times $p_j$ if one allows \emph{preemption}.

Now, our setting is a little bit different. For each time $t$, we have a certain number $\textrm{cap}(t) \in \{ 0,\ldots,m\}$ of slots. Additionally, we have a precedence order that we need to respect. 
But we can use to our advantage that the precedence order has only short chains; moreover, the number of different release times and deadlines is limited. 

Formally we will prove the following:
\begin{theorem}\label{thm:EDF}
Let $J$ be a set of jobs with release times $r_j$, deadlines $d_j$ and consistent
precedence constraints\footnote{Here ``consistent'' means that 
for a pair of dependent jobs $j \prec j'$ one has $r_{j} \leq r_{j'}$ and $d_j \leq d_{j'}$.} with maximum chain length $C$.
Suppose that $\{1,\ldots,T\}$ is the time horizon, partitioned into
equally sized intervals $I_1,\ldots,I_p$
and all release times/deadlines correspond to beginnings and ends of
those intervals. Let
$\textrm{cap} : [T] \to \{ 0,\ldots,m\}$
be a capacity function and assume that there exists a 
schedule $\tilde{\sigma} : J \to [T]$ assigning each job to slots between its release time and deadline that respects
capacities but does not necessarily respect precedence constraints within $J$.

Then in polynomial time, one can find a schedule $\sigma : J \setminus J_{\textrm{discarded}} \to [T]$ that respects
capacities, release times, deadlines and precedence
constraints and discards  $|J_{\textrm{discarded}}| \leq p^2mC$ many jobs.
\end{theorem}

We use the following algorithm, which is a variant of Earliest Deadline First, where we
discard those jobs that we cannot process in time:
\begin{center}
\psframebox{
\begin{minipage}{14.0cm}
{\sc Earliest Deadline First} \vspace{1mm} \hrule  \vspace{1mm}
{\bf Input:} Jobs $J$ with deadlines, release times, precedence constraints; capacity function $\textrm{cap} : [T] \to \{ 0,\ldots,m\}$ \\
{\bf Output:} Schedule $\sigma : J \to [T] \cup \{ \texttt{DISCARDED}\}$ \vspace{1mm} \hrule \vspace{2mm} 
\begin{enumerate*}
\item[(1)] Set $\sigma(j) := \texttt{UNASSIGNED}$ for all $j \in J$ and
$J_{\textrm{discarded}} := \emptyset$
\item[(2)] Sort the jobs $J = \{ 1,\ldots,n\}$ so that $d_1 \leq d_2
\leq \ldots \leq d_n$
\item[(3)] FOR $t=1$ TO $T$ DO \vspace{-1.0mm}
   \begin{enumerate*}
   \item[(4)] FOR $\textrm{cap}(t)$ MANY TIMES DO
   \begin{enumerate*}
     \item[(5)] Select the lowest index $j$ of a job with $r_j \leq t
\leq d_j$ that has not been scheduled or discarded
and that has all jobs in $\delta^-(j)$ already processed (or discarded).
     \item[(6)] Set $\sigma(j) := t$ (if there was any such job)
   \end{enumerate*}
   \item[(7)] FOR each  $j \in J$ with $d_j=t$ and $\sigma(j) =
\texttt{UNASSIGNED}$, add $j$ to $J_{\textrm{discarded}}$
and set $\sigma(j) := \texttt{DISCARDED}$
   \end{enumerate*}
\end{enumerate*}
\end{minipage}
}
\end{center}
At the end all jobs $j$ will be either scheduled between $r_j$ and $d_j$
(that means $r_j \leq \sigma(j) \leq d_j$) or they are
in $J_{\textrm{discarded}}$.

We will say that a job $j$ was \emph{discarded in the interval $[t,t']$}
if $j \in J_{\textrm{discarded}}$ and
$t \leq d_j \leq t'$. We call a time $t$ \emph{busy} if
$|\sigma^{-1}(t)| = \textrm{cap}(t)$ and \emph{non-busy} otherwise.
Let us make a useful observation:
\begin{lemma} \label{lem:dep-discarded-job-to-nonbusy-time}
Let $I=\{t',\ldots,t''\} \subseteq I_i$ be part of one of the subintervals. 
Suppose that there is a non-busy time $t^* \in I$
and a job $j$ with $I \subseteq \{r_j,\ldots,d_j\}$ and $\sigma(j) \in \{
\texttt{DISCARDED} \} \cup \{ t''+1,\ldots,T\}$. Then there is a job
$j^* \in \sigma^{-1}(t^*)$
with $j^* \prec j$.
\end{lemma}
\begin{proof}
Consider any inclusion-wise maximal chain of jobs $j_1 \prec j_2 \prec
\ldots \prec j_q$ that
ends in $j = j_q$ and otherwise has only jobs $j_1,\ldots,j_{q-1} \in
\sigma^{-1}(\{t^*,\ldots,t''\})$.
First suppose that $q>1$ and hence $j_1 \neq j$.
It is impossible that $\sigma(j_1) > t^*$ because by assumption $r_{j_1}
\leq t^*$ and hence
EDF would have processed $j_1$ already earlier at time $t^*$ (by
maximality of the chain, there is
no job scheduled at times $t^*,\ldots,\sigma(j_1)-1$ on which $j_1$
might depend).
Hence  $\sigma(j_1) = t^*$ and by transitivity $j_1 \prec j$, which
gives the claim.

In the 2nd case, we have $j_1 = j$, hence there is no job that $j$
depends on scheduled between $t^*$ and $t''$.
But we know that $r_{j} \leq t^*$. 
Thus EDF would have processed $j$ at time $t^*$ or earlier.
\end{proof}

With this observation we can easily limit the number of non-busy times
per interval:

\begin{lemma} \label{lem:Bound-on-num-nonbusy-times}
Let $I = \{ t',\ldots,t''\} \subseteq I_i$ be part of one of the subintervals. 
Suppose that there is at least one job $j$ with $I \subseteq \{r_j,\ldots,d_j\}$ and $\sigma(j) \in \{
\texttt{DISCARDED}\} \cup \{ t''+1,\ldots,T\}$.
Then the number of non-busy times in $I$ is bounded by $C$.
\end{lemma}
\begin{proof}
By Lemma~\ref{lem:dep-discarded-job-to-nonbusy-time}, for the latest time $t^*
\in I$
with $|\sigma^{-1}(t^*)| < \textrm{cap}(t^*)$, there is at least one job $j^*
\in \sigma^{-1}(t^*)$ with $j^* \prec j$. 
Then we can continue by induction, replacing $t''$ with $t^*-1$ and replacing $j$ by $j^*$
to build a chain of jobs ending with $j$ that includes a job scheduled at each non-busy time.
Since no chains can be longer than $C$, this gives the claim.
\end{proof}

Now we come to the main argument where we give an upper bound of the
number of discarded jobs:
\begin{lemma} 
One has $|J_{\textrm{discarded}}| \leq p^2mC$.
\end{lemma}

\begin{proof}
Suppose that $|J_{\textrm{discarded}}| \geq p\cdot K$; we will
then derive
a bound on $K$.
By the pigeonhole principle, we can find an interval $I_b$
so that
at least $K$ many jobs get discarded in $I_b$.
Let us denote the lowest priority (i.e. the latest deadline) job
that gets discarded in $I_b$ by $j_s$. Now delete all those lower priority
jobs $j_{s+1},\ldots,j_n$. Note that this does not affect how EDF schedules
$j_1,\ldots,j_s$ and still we would have at least $K$ jobs discarded in
$I_b$, including $j_s$.
By Lemma~\ref{lem:Bound-on-num-nonbusy-times},
the number of non-busy periods in $I_b$
is bounded by $C$. Now, choose a minimal index $a \in \{
1,\ldots,b-1\}$ so that in
each of the
intervals $I_{a},\ldots,I_b$ one has at most $C$ many non-busy
periods.
We abbreviate $I' := I_{a} \cup \ldots \cup I_b$. Note that
by definition $I_{a-1}$ has more than $C$ many non-busy
periods\footnote{Admittedly it is possible that $a=1$ in which
case one might imagine $I_0$ as an interval in which all times are
non-busy and which does not contain any release times.}.
Define
\[
J' := \Big\{ j \in \{ j_1,\ldots,j_s\} : (\sigma(j) \in I') \textrm{ or
} \\ (j\textrm{ discarded and }d_j \in I')  \Big\}.
\]
By Lemma~\ref{lem:dep-discarded-job-to-nonbusy-time}, any job in $J'$ has
its release time in $I_a$ or later, since otherwise we could not have
$C+1$
non-busy times in $I_{a-1}$. Now, let us double count the number of
jobs in $J'$.
On the one hand, we have
\begin{eqnarray*}
   |J'| \geq \sum_{i=a}^b |\sigma^{-1}(I_i)| + |J' \cap
J_{\textrm{discarded}}| 
&\geq& \sum_{i=a}^b (\textrm{cap}(I_i) - mC) + \underbrace{|J' \cap
J_{\textrm{discarded}}|}_{\geq K} \\
&\geq& \textrm{cap}(I') - pmC + K.
\end{eqnarray*}
On the other hand, we know that there is an assignment $\tilde{\sigma}$ of jobs
in $J'$ to slots in $I_{a},\ldots,I_{b}$. That tells us that $|J'| \leq \textrm{cap}(I')$.
Comparing both bounds gives that $K \leq pmC$.
\end{proof}




\section{Step (2) --- Accounting for Discarded Jobs}\label{sec:AccountingDiscardedJobs}

In this section, we need to argue that the level $\ell^*$ can be chosen so that the total 
number of jobs that are discarded in Steps (1)-(4) are bounded by
\[
 |J_{\textrm{discarded}}^*| \leq \frac{\varepsilon}{2} \cdot \frac{\log (|I^*|)}{\log(T)} \cdot |I^*| + \frac{\varepsilon}{2m} \cdot |J^*|
\]
as claimed. 
Let us summarize the three occasions in the algorithm where a job might get discarded: 
\begin{enumerate}
\item[(A)] In Step (3), in order to schedule the bottom jobs, 
we have $2^{\ell^*}$ many recursive calls of Lemma~\ref{lem:MainRecursiveAlgorithm}
for intervals $I \in \pazocal{I}_{\ell^*}^*$. The cumulative number of discarded jobs from all those calls
is bounded by
\[
 2^{\ell^*} \cdot \frac{\varepsilon}{2} \cdot \frac{\log(\frac{T^*}{2^{\ell^*}})}{\log(T)} \cdot \frac{T^*}{2^{\ell^*}} + \frac{\varepsilon}{2m} \cdot |J_{\textrm{bottom}}| 
= \frac{\varepsilon}{2} \cdot \frac{\log(T^*)-\ell^*}{\log(T)} \cdot T^* + \frac{\varepsilon}{2m} \cdot |J_{\textrm{bottom}}|
\]
\item[(B)] As we have seen in Section~\ref{sec:SchedulingTopJobs}, the number of top jobs that 
need to be discarded in Step (4) can be bounded by 
\[
4m \cdot 2^{-k} \cdot T^* + p^2mC 
 \leq 4m \cdot 2^{-k} \cdot T^* + k^2 m2^{2k^2} T^* \delta 
\leq \frac{\varepsilon}{4} \cdot \frac{1}{\log(T)} \cdot T^*.
\]
Here we use that the length of the maximum chain within the top jobs is $C \leq k^2 \delta T^*$.
Moreover, we have substituted the parameters $p = 2^{\ell^*} \leq 2^{k^2}$
as well as  $\delta = \frac{\varepsilon}{8k^2 m2^{2k^2} \log(T)}$ and $k = c_1 \frac{m}{\varepsilon} \log \log(T)$ with a large enough constant $c_1>0$.
\item[(C)] In Step (2), we discard all the middle jobs. In the remainder of this section we prove 
that there is an index $\ell^*$ so that 
\[
   |J_{\textrm{middle}}| \leq \frac{\varepsilon}{4} \cdot \frac{1}{\log(T)} \cdot T^* + \frac{\varepsilon}{2m} \cdot (|J_{\textrm{middle}}| + |J_{\textrm{top}}|)
\]
\end{enumerate}

Let us first assume that we can indeed find a proper index $\ell^*$ so that the bound in (C)
is justified. Then the total number of jobs that the algorithm discards
is
\begin{eqnarray*}
& & \overbrace{\frac{\varepsilon}{2} \cdot \frac{\log(T^*)-\ell^*}{\log(T)} \cdot T^* + \frac{\varepsilon}{2m} \cdot |J_{\textrm{bottom}}|}^{(A)} 
 + \overbrace{\frac{\varepsilon}{4} \cdot \frac{1}{\log(T)} \cdot T^*}^{(B)} \\
& & +\overbrace{\frac{\varepsilon}{4} \cdot \frac{1}{\log(T)} \cdot T^* + \frac{\varepsilon}{2m} \cdot (|J_{\textrm{middle}}| + |J_{\textrm{top}}|)}^{(C)} \\
&\leq& \frac{\varepsilon}{2} \cdot \frac{\log(T^*)}{\log(T)}\cdot T^* + \frac{\varepsilon}{2m} \cdot \underbrace{(|J_{\textrm{top}}| + |J_{\textrm{middle}}| + |J_{\textrm{bottom}}|)}_{=|J^*|}
\end{eqnarray*}
which is the bound that we claimed in Lemma~\ref{lem:MainRecursiveAlgorithm}. 

It remains to justify the claim in (C). 
We abbreviate   $\alpha_{i} := |J(i \cdot k,\bm{x}^{**}) \cup \ldots \cup J((i+1) \cdot k-1,\bm{x}^{**})|$ for  $i \in \{ 0,\ldots,k-1\}$. 
In words, each number $\alpha_i$ represents the number of jobs owned by $k$ consecutive levels. 
We observe that if there is an index $i \in \{ 0,\ldots,k-1\}$ so that 
\[
(I) \quad \alpha_i \leq \frac{\varepsilon}{4\log(T)} \cdot T^* \quad\quad \textrm{or} \quad\quad (II) \quad \alpha_i \leq \frac{\varepsilon}{2m} \cdot \sum_{j=1}^i \alpha_j
\]
then we can choose $\ell^* := (i+1) \cdot k$ and (C) will be satisfied. Here we use that for this particular choice of $\ell^*$, 
we will have $|J_{\textrm{middle}}| = \alpha_i$ and $|J_{\textrm{top}}| = \alpha_0 + \ldots + \alpha_{i-1}$.

So, we assume for the sake of contradiction that no index $i$ satisfies either $(I)$ or $(II)$ (or both). 
Then one can easily show that the $\alpha_i$'s have to grow exponentially. We show this in a small lemma:
\begin{lemma} \label{lem:GrowOfSequenceAlphaIs}
Let  $q \in \setN$ and suppose we have a sequence of numbers $\alpha_0,\alpha_1,\ldots,\alpha_{N}$ 
satisfying $\alpha_i \geq \alpha_{\min}>0$ and $\alpha_i \geq \frac{1}{q} \cdot \sum_{j=1}^i \alpha_j$ for all $i=0,\ldots,N$. 
Then $\alpha_i \geq 2^{\lfloor i/(2q) \rfloor} \cdot \alpha_{\min}$. 
\end{lemma}
\begin{proof}
Group the indices into consecutive \emph{blocks} of $2q$ numbers, where $\alpha_0,\ldots,\alpha_{2q-1}$ is 
block $0$, $\alpha_{2q},\ldots,\alpha_{4q-1}$ is block 1 and so on. We want to prove by induction that
each $\alpha_i$ in the $j$th block is at least $2^{j} \cdot \alpha_{\min}$. For $j=0$, the claim follows from the 
assumption. For $j>0$, we use that $\alpha_i$ is at least the sum of $2q$ numbers 
that by inductive hypothesis are all at least $\frac{2^{j-1}}{q} \cdot \alpha_{\min}$. The claim follows. 
\end{proof}

Applying Lemma~\ref{lem:GrowOfSequenceAlphaIs} with $\alpha_{\min} := \frac{\varepsilon}{4\log(T)} \cdot T^*$ and
$q := \frac{2m}{\varepsilon}$ we obtain in particular that 
\[
  \alpha_{k-1} \geq 2^{\lfloor (k-1) \frac{\varepsilon}{4m} \rfloor} \cdot \frac{\varepsilon}{4\log(T)} \cdot T^*
\]
If we set  
$k =\frac{c_1m}{\varepsilon} \log(\log(T))$ for some adequately large $c_1$, then
 $\alpha_{k-1} > mT^*$, which is a contradiction
since we only have $|J^*| \leq m \cdot |I^*| = mT^*$ many jobs with $x_{j,I^*} = 1$. 

\section{Conclusion} 

For the proof of Lemma~\ref{lem:MainRecursiveAlgorithm} we already argued that the
number of discarded jobs satisfies the claimed bound and that all precedence constraints
will be satisfied. Regarding the number of rounds in the hierarchy, recall that
we started with a solution $\bm{x}^* \in \texttt{SA}(K(T),r^*)$ with $r^* \geq \log(T^*) \cdot 2mk^2 \cdot 2^{k^2} / \delta$. Then we apply a round of conditionings in Lemma~\ref{lem:BreakingChains} 
to obtain $\bm{x}^{**} \in \texttt{SA}(K(T),r^{**})$ with $r^{**} := r^* -2mk^2 \cdot 2^{k^2} / \delta$.
We use copies of the solution $\bm{x}^{**} $ in our recursive application of Lemma~\ref{lem:MainRecursiveAlgorithm} to intervals of size $T^* / 2^{\ell^*}$. Since $\ell^* \geq k \geq 1$, the remaining 
number of 
LP-hierarchy rounds satisfies $r^{**} \geq \log(T^* / 2^{\ell^*}) \cdot 2mk^2 \cdot 2^{k^2} / \delta$. Thus we still have enough 
LP-hierarchy rounds for the recursion.  

Another remark concerns why we may assume that the precedence constraints are consistent in 
Theorem~\ref{thm:EDF}. 
Suppose we have jobs $j,j' \in J_{\textrm{top}}$ with $j \prec j'$ and consider the definition of release times and
deadlines in Eq.~\ref{eq:ReleaseTimeDeadlineDef}. By transitivity, any job $j'' \in J_{\textrm{bottom}}$ with $j' \prec j''$ which limits the deadline
of $j'$ will also limit the deadline of $j$, hence $d_j \leq d_{j'}$. Similarly one can argue
that $r_{j} \leq r_{j'}$.
This concludes the proof of Lemma~\ref{lem:MainRecursiveAlgorithm}
and our main result follows. 

\section{Follow-up work and open problems}

A natural question that arises is whether the number of $O(\log n)^{O(\log \log n)}$
rounds for constant $\varepsilon,m$ can be improved. In fact, after the conference version of this
paper appeared, Garg~\cite{QPTAS-Scheduling-GargArxiv2017} was able to reduce the number of rounds down 
to $O(\log n)^{O(1)}$, hence providing an actual QPTAS. It remains open whether $c(m,\varepsilon)$ many rounds
suffice as well. Another tantalizing question is whether a similar approach could give a $(1+\varepsilon)$-approximation for $Pm \mid \textrm{prec} \mid C_{\max}$, where the processing times $p_j \in \setN$ are arbitrary. Note that the
difficulty in this setting comes from the issue that jobs have to be scheduled non-preemptively.

\bibliographystyle{alpha}
\bibliography{QPTASPrecScheduling}

\begin{thebibliography}{GTW13}

\bibitem[ABS10]{UniqueGamesAlgo-AroraBarakSteurer-FOCS10}
S.~Arora, B.~Barak, and D.~Steurer.
\newblock Subexponential algorithms for unique games and related problems.
\newblock In {\em 51th Annual {IEEE} Symposium on Foundations of Computer
  Science, {FOCS} 2010, October 23-26, 2010, Las Vegas, Nevada, {USA}}, pages
  563--572, 2010.

\bibitem[ARV09]{SparsestCut-AroraRaoVazirani-JACM09}
S.~Arora, S.~Rao, and U.~V. Vazirani.
\newblock Expander flows, geometric embeddings and graph partitioning.
\newblock {\em J. {ACM}}, 56(2), 2009.

\bibitem[BK09]{OptLongCodeTest-BansalKhot-FOCS09}
N.~Bansal and S.~Khot.
\newblock Optimal long code test with one free bit.
\newblock In {\em Foundations of Computer Science, 2009. FOCS '09. 50th Annual
  IEEE Symposium on}, pages 453--462, Oct 2009.

\bibitem[BSS16]{LiftAndRound-BansalSrinivasanSvensson-STOC2016}
N.~Bansal, S.~Srinivasan, and O.~Svensson.
\newblock Lift-and-round to improve weighted completion time on unrelated
  machines.
\newblock In {\em Proceedings of the 48th Annual {ACM} {SIGACT} Symposium on
  Theory of Computing, {STOC} 2016, Cambridge, MA, USA, June 18-21, 2016},
  pages 156--167, 2016.

\bibitem[Cha95]{TechReportCharikar1995}
M.~Charikar.
\newblock Approximation algorithms for problems in combinatorial optimization.
\newblock Technical report, B. Tech. Project Report, Department of Computer Sc.
  and Engg., IIT Bombay, 1995.

\bibitem[Chl07]{ApproxAlgoViaSDP-Chlamtac-FOCS07}
E.~Chlamtac.
\newblock Approximation algorithms using hierarchies of semidefinite
  programming relaxations.
\newblock In {\em 48th Annual {IEEE} Symposium on Foundations of Computer
  Science {(FOCS} 2007), October 20-23, 2007, Providence, RI, USA,
  Proceedings}, pages 691--701, 2007.

\bibitem[CKR10]{SparsestCutInBoundedTreeWidthGraphs-CKR-APPROX2010}
E.~Chlamtac, R.~Krauthgamer, and P.~Raghavendra.
\newblock Approximating sparsest cut in graphs of bounded treewidth.
\newblock In {\em Approximation, Randomization, and Combinatorial Optimization.
  Algorithms and Techniques, 13th International Workshop, {APPROX} 2010, and
  14th International Workshop, {RANDOM} 2010, Barcelona, Spain, September 1-3,
  2010. Proceedings}, pages 124--137, 2010.

\bibitem[CS08]{HypergraphColoringChlamtacSinghAPPROX08}
E.~Chlamtac and G.~Singh.
\newblock Improved approximation guarantees through higher levels of {SDP}
  hierarchies.
\newblock In {\em Approximation, Randomization and Combinatorial Optimization.
  Algorithms and Techniques, 11th International Workshop, {APPROX} 2008, and
  12th International Workshop, {RANDOM} 2008, Boston, MA, USA, August 25-27,
  2008. Proceedings}, pages 49--62, 2008.

\bibitem[Dag10]{DagstuhlOpenProblems2010}
Open problems -- scheduling.
\newblock In Susanne Albers, Sanjoy~K. Baruah, Rolf~H. M{\"o}hring, and Kirk
  Pruhs, editors, {\em Scheduling}, number 10071 in Dagstuhl Seminar
  Proceedings, Dagstuhl, Germany, 2010. Schloss Dagstuhl - Leibniz-Zentrum fuer
  Informatik, Germany.

\bibitem[Der74]{EDF-Dertouzos74}
M.~L. Dertouzos.
\newblock Control robotics: The procedural control of physical processes.
\newblock In {\em {IFIP} Congress}, pages 807--813, 1974.

\bibitem[Gar17]{QPTAS-Scheduling-GargArxiv2017}
S.~Garg.
\newblock Quasi-ptas for scheduling with precedences using {LP} hierarchies.
\newblock {\em CoRR}, abs/1708.04369, 2017.

\bibitem[GJ79]{GareyJohnson79}
M.~R. Garey and D.~S. Johnson.
\newblock {\em Computers and Intractability: {A} Guide to the Theory of
  {NP}-Completeness}.
\newblock W. H. Freeman and Company, New York, New York, 1979.

\bibitem[GR08]{PrecedenceConstrainedScheduling-GangalRanade2008}
D.~Gangal and A.~Ranade.
\newblock Precedence constrained scheduling in optimal.
\newblock {\em Journal of Computer and System Sciences}, 74(7):1139 -- 1146,
  2008.

\bibitem[Gra66]{Graham66}
R.~L. Graham.
\newblock Bounds for certain multiprocessing anomalies.
\newblock {\em Bell System Technical Journal}, 45(9):1563--1581, 1966.

\bibitem[GTW13]{SparsestCut-GuptaTalwarWitmer-STOC2013}
A.~Gupta, K.~Talwar, and D.~Witmer.
\newblock Sparsest cut on bounded treewidth graphs: algorithms and hardness
  results.
\newblock In {\em Symposium on Theory of Computing Conference, STOC'13, Palo
  Alto, CA, USA, June 1-4, 2013}, pages 281--290, 2013.

\bibitem[GW95]{MaxCut-GoemansWilliamson-JACM95}
M.~X. Goemans and D.~P. Williamson.
\newblock Improved approximation algorithms for maximum cut and satisfiability
  problems using semidefinite programming.
\newblock {\em J. {ACM}}, 42(6):1115--1145, 1995.

\bibitem[LK78]{ComplexityOfScheduling-Lenstra-RinnoyKan1978}
J.~K. Lenstra and A.~H. G.~Rinnooy Kan.
\newblock Complexity of scheduling under precedence constraints.
\newblock {\em Operations Research}, 26(1):22--35, 1978.

\bibitem[LS77]{TwoSchedulingAlgorithms-Lam-Sethi-SICOMP77}
S.~Lam and R.~Sethi.
\newblock Worst case analysis of two scheduling algorithms.
\newblock {\em {SIAM} J. Comput.}, 6(3):518--536, 1977.

\bibitem[RT12]{CSPs-with-card-constraints-RaghavendraTanSODA12}
P.~Raghavendra and N.~Tan.
\newblock Approximating csps with global cardinality constraints using {SDP}
  hierarchies.
\newblock In {\em Proceedings of the Twenty-Third Annual {ACM-SIAM} Symposium
  on Discrete Algorithms, {SODA} 2012, Kyoto, Japan, January 17-19, 2012},
  pages 373--387, 2012.

\bibitem[Sch03]{CombinatorialOptimizationBook-Schrijver-2003}
A.~Schrijver.
\newblock {\em Combinatorial Optimization - Polyhedra and Efficiency}.
\newblock Springer, 2003.

\bibitem[Sve10]{HardnessPrecedenceScheduling-Svensson-STOC2010}
O.~Svensson.
\newblock Conditional hardness of precedence constrained scheduling on
  identical machines.
\newblock In {\em Proceedings of the Forty-second ACM Symposium on Theory of
  Computing}, STOC '10, pages 745--754, New York, NY, USA, 2010. ACM.

\bibitem[SW99]{TenOpenProblems-SchuurmanWoeginger1991}
P.~Schuurman and G.~J. Woeginger.
\newblock Polynomial time approximation algorithms for machine scheduling: ten
  open problems.
\newblock {\em Journal of Scheduling}, 2(5):203--213, 1999.

\end{thebibliography}

\end{document}